\documentclass[submission,copyright,creativecommons]{eptcs} 
\providecommand{\event}{SYNT 2012}
\usepackage[T1]{fontenc}
\usepackage[latin9]{inputenc}
\usepackage{verbatim}
\usepackage{float}
\usepackage{wrapfig}
\usepackage{amsmath}
\usepackage{amsthm}
\usepackage{amssymb}
\usepackage{graphicx}
\PassOptionsToPackage{normalem}{ulem}
\usepackage{ulem}

\makeatletter

\floatstyle{ruled}
\newfloat{algorithm}{tbp}{loa}
\providecommand{\algorithmname}{Algorithm}
\floatname{algorithm}{\protect\algorithmname}

\numberwithin{equation}{section}
\numberwithin{figure}{section}
\theoremstyle{plain}
\newtheorem{thm}{\protect\theoremname}
  \theoremstyle{definition}
  \newtheorem{example}[thm]{\protect\examplename}
  \theoremstyle{plain}
  \newtheorem{prop}[thm]{\protect\propositionname}
\newenvironment{lyxcode}
{\par\begin{list}{}{
\setlength{\rightmargin}{\leftmargin}
\setlength{\listparindent}{0pt}
\raggedright
\setlength{\itemsep}{0pt}
\setlength{\parsep}{0pt}
\normalfont\ttfamily}%
 \item[]}
{\end{list}}

\def\concat{{+\mkern-4mu +}} 

\makeatother

\usepackage[english]{babel}
  \providecommand{\examplename}{Example}
  \providecommand{\propositionname}{Proposition}
\providecommand{\theoremname}{Theorem}

\title{Theory and Techniques for Synthesizing a Family of Graph Algorithms}

\author{Srinivas Nedunuri \quad \quad William R. Cook \institute{University of Texas at Austin} \email{nedunuri|wcook@cs.utexas.edu} \and   Douglas R. Smith \institute{Kestrel Institute} \email{smith@kestrel.edu}}
\def\titlerunning{Theory and Techniques for Synthesizing a Family of Graph Algorithms}
\def\authorrunning{S.Nedunuri, W. R. Cook \& D.R. Smith}
\begin{document}

\maketitle

\begin{abstract}
Although Breadth-First Search (BFS) has several advantages over Depth-First
Search (DFS) its prohibitive space requirements have meant that algorithm
designers often pass it over in favor of DFS. To address this shortcoming,
we introduce a theory of Efficient BFS (EBFS) along with a simple
recursive program schema for carrying out the search. The theory is
based on dominance relations, a long standing technique from the field
of search algorithms. We show how the theory can be used to systematically
derive solutions to two graph algorithms, namely the Single Source
Shortest Path problem and the Minimum Spanning Tree problem. The solutions
are found by making small systematic changes to the derivation, revealing
the connections between the two problems which are often obscured
in textbook presentations of them. 
\end{abstract}

\section{Introduction}

Program synthesis is experiencing something of a resurgence \cite{SGF10,STBSS06,GJTV11}
\cite{PBS11,VY08} following negative perceptions of its scalability
in the early 90s. Many of the current approaches aim for near-automated
synthesis. In contrast, the approach we follow, we call \emph{guided
program synthesis}, also incorporates a high degree of automation
but is more user-guided. The basic idea is to identify interesting
classes of algorithms and capture as much \emph{generic} algorithm
design knowledge as possible in one place.The user instantiates that
knowledge with problem-specific \emph{domain} information. This step
is often carried out with machine assistance. The approach has been
applied to successfully derive scores of efficient algorithms for
a wide range of practical problems including scheduling \cite{SPW95},
concurrent garbage collection \cite{PPS10}, and SAT solvers \cite{SW08}.

One significant class of algorithms that has been investigated is
search algorithms. Many interesting problems%
{} can be solved by application of search. In such an approach, an initial
search space is partitioned into subspaces, a process called \emph{splitting,}
which continues recursively until a \emph{feasible} solution is found.
A feasible solution is one that satisfies the given problem specification.
Viewed as a search tree, spaces form nodes, and the subspaces after
a split form the children of that node. %
The process%
{} has been formalized by Smith \cite{Smi88,Smi10}. Problems which
can be solved by global search are said to be in the Global Search
(GS) class. The enhancements in GS over standard branch-and-bound
include a number of techniques designed to improve the quality of
the search by eliminating unpromising avenues. One such technique
is referred to as \emph{dominance relations}. Although they do not
appear to have been widely used, the idea of dominance relations goes
back to at least the 70s \cite{Iba77}. Essentially, a dominance relation
is a relation between two nodes in the search tree such that if one
dominates the other, then the dominated node is guaranteed to lead
to a worse solution than the dominating one, and can therefore be
discarded. Establishing a dominance relation for a given problem is
carried out by a user. However this process is not always obvious.
There are also a variety of ways in which to carry out the search,
for example Depth-First (DFS), Breadth-First (BFS), Best-First, etc.
Although DFS is the most common, BFS actually has several advantages
over DFS were it not for its exponential space requirement. The key
to carrying out BFS space-efficiently is to limit the size of the
frontier at any level. However, this has not been investigated in
any systematic manner up to now.

This paper has two main contributions:
\begin{itemize}
\item We show how to limit the size of the frontier in search using dominance
relations, thereby enabling space-efficient BFS. From this formal
characterization, we derive a characteristic recurrence that serves
as the basis of a program schema for implementing Global Search. Additionally,
we show that limiting the size of the undominated frontier to one
results in a useful class of \emph{greedy} algorithms.
\item We show how to derive dominance relations and demonstrate they satisfy
the greediness conditions for two graph problems, namely Single Source
Shortest Path and Minimum Spanning Tree by a systematic process, which
though not automatic, we believe has the potential to be automated.
\end{itemize}

\section{Background To Guided Program Synthesis}

\subsection{\label{sub:Process-1}Process}

The basic steps in guided program synthesis are:
\begin{enumerate}
\item Start with a logical specification of the problem to be solved. A
specification is a quadruple $\langle D,R,o,c\rangle$ where $D$
is an input type, $R$ an output or result type, $o:D\times R$ is
a predicate relating correct or feasible outputs to inputs, and $c:D\times R\rightarrow Int$
is a cost function on solutions. An example specification is in Eg.
\ref{Spec-of-SP} (This specification is explained in more detail
below)
\item Pick an algorithm class from a library of algorithm classes (\textsc{Global
Search, Local Search, Divide and Conquer, Fixpoint Iteration}, etc).
An algorithm class comprises a \emph{program schema} containing operators
to be instantiated and an \emph{axiomatic theory} of those operators
(see \cite{Ned12} for details). A schema is analogous to a template
in Java/C++ %
, with the difference that both the template and template arguments
are formally constrained.
\item Instantiate the operators of the program schema using information
about the problem domain and in accordance with the axioms of the
class theory. To ensure correctness, this step can be carried out
with mechanical assistance. %
The result is an efficient algorithm for solving the given problem. 
\item Apply low-level program transforms such as finite differencing, context-dependent
simplification, and partial evaluation, followed by code generation.
Many of these are automatically applied by Specware \cite{S}, a formal
program development environment.
\end{enumerate}
The result of Step 4 is an efficient program for solving the problem
which is guaranteed correct by construction. The power of the approach
stems from the fact that the common structure of many algorithms is
contained in \emph{one} reusable program schema and associated theory.
Of course the program schema needs to be carefully designed, but that
is done once by the library designer. The focus of this %
{} paper is the \textsc{Global Search} class, and specifically on how
to methodically carry out Step 3 for a wide variety of problems. Details
of the other algorithm classes and steps are available elsewhere \cite{Kre98,Smi88,PPS10}.
\begin{example}
\label{Spec-of-SP}Specification of the Single Pair Shortest Path
(SPSP) problem is shown in Fig. \ref{fig:Spec-of-SP-1} (The $\mapsto$
reads as ``instantiates to'')%
{} The input $D$ is a structure with 3 fields, namely a start node,
end node and a set of edges. The result $R$ is a sequence of edges
($[\ldots]$ notation). A correct result is one that satisfies the
predicate $path?$ which checks that a path $z$ must be a contiguous
path from the start node to the end node ( simple recursive definition
not shown). Finally the cost of a solution is the sum of the costs
of the edges in that solution. Note that fields of a structure are
accessed using the '.' notation. 
\end{example}
\begin{wrapfigure}{o}{0.5\columnwidth}%
$\begin{array}{rcl}
D & \mapsto & \langle s:Node,e:Node,edges:\{Edge\}\rangle\\
 &  & Edge=\langle f:Node,t:Node,w:Nat\rangle\\
R & \mapsto & [Edge]\\
o & \mapsto & \lambda(x,z)\cdot\, path?(z,x.s,x.e)\\
 &  & path?(p,s,f)=...\\
c & \mapsto & \lambda(x,z)\,\cdot\sum_{edge\in z}edge.w
\end{array}$

\caption{\label{fig:Spec-of-SP-1}Specification of Shortest Path problem}
\end{wrapfigure}%

\subsection{Global Search}

Before delving into a program schema for Global Search, it helps to
understand the structures over which the program schema operates.
In \cite{Smi88}, a \emph{search space} is represented by a descriptor
of some type $\widehat{R}$, which is an abstraction of the result
type $R$. %
{} %
The initial or starting space is denoted $\bot$. There are also two
predicates \emph{split$:D\times\widehat{R}\times\widehat{R}$,} written
$\pitchfork$, and \emph{extract}$:\widehat{R}\times R$, written
$\chi$. Split defines when a space is a subspace of another space,
and extract captures when a solution is extractable from a space%
. We say a solution $z$ is \emph{contained} in a space $y$ (written
$z\in y$) if it can be extracted after a finite number of splits.
A feasible space is one that contains feasible solutions. We often
write $\pitchfork(x,y,y')$ as $y\pitchfork_{x}y'$ for readability,
and even drop the subscript when there is no confusion. \emph{Global
Search theory (GS-theory)} \cite{Smi88} axiomatically characterizes
the relation between the predicates $\bot$, $\pitchfork$ and $\chi$,
as well as ensuring that the associated program schema computes a
result that satisfies the specification. In the sequel, the symbols
$\widehat{R},\bot,\pitchfork,\chi,\oplus$ are all assumed to be drawn
from GS-theory. A theory for a given problem is created by instantiating
these terms, as shown in the next example. 
\begin{wrapfigure}{l}{0.5\columnwidth}%
$\begin{array}{rcl}
\widehat{R} & \mapsto & R\\
\bot & \mapsto & \lambda x\cdot\,[]\\
\pitchfork & \mapsto & \lambda(x,p,p')\cdot\,\exists e\in x.edges\cdot\\
 & & \,\,\,\,\,\,\,\,\,\,\,\,\,\,\,\,\,\,\,\,\,\,\,\,\,\,\,\,\,\,\,\,\,\,\,\,\,\,\,\,\,\, p'=p\concat[e]\\
\chi & \mapsto & \lambda(z,p)\cdot\, p=z
\end{array}$\caption{\label{fig:GS-for-SP-1}GS instantiation for Single Pair Shortest
Path}
\end{wrapfigure}%

\begin{example}
\label{GS-theory-for-SP}Instantiating GS-theory for the Single Pair
Shortest Path problem. The type of solution spaces $\widehat{R}$
is the same as the result type $R$ %
\footnote{there is a covariant relationship between an element of $\widehat{R}$
and of $R$. For example, the initial space, corresponding to all
possible paths, is the empty list.%
}. A space is split by adding an edge to the current path - that is
the subspaces are the different paths that result from adding an edge
to the parent path. Finally a solution can be trivially extracted
from any space by setting the result $z$ to the space $p$. This
is summarized in Fig. \ref{fig:GS-for-SP-1} ($[]$ denotes %
{} the empty list, and $\concat$ denotes list concatenation).
\end{example}

\subsection{Dominance Relations}

As mentioned in the introduction, a dominance relation provides a
way of comparing two subspaces in order to show that one will always
contain at least as good a solution as the other. (Goodness in this
case is measured by some cost function on solutions). The first space
is said to \emph{dominate} ($\vartriangleright$) the second, which
can then be eliminated from the search. Letting $c^{*}$ denote the
cost of an optimal solution in a space, this can be formalized as
(all free variables are assumed to be universally quantified):
\begin{equation}
y\vartriangleright y'\Rightarrow c^{*}(x,y)\leq c^{*}(x,y')\label{eq:dominance}
\end{equation}
Another way of expressing the consequent of (\ref{eq:dominance})
is
\begin{equation}
\forall z'\in y'\cdot\, o(x,z')\Rightarrow\exists z\in y\cdot\, o(x,z)\wedge c(x,z)\leq c(x,z')\label{eq:general-dominance}
\end{equation}

{} To derive dominance relations, it is often useful to first derive
a semi-congruence relation \cite{Smi88}. A semi-congruence between
two partial solutions $y$ and $y'$, written $y\rightsquigarrow y'$,
ensures that any way of extending $y'$ into a feasible solution can
also be used to extend $y$ into a feasible solution. Like $\pitchfork$,
$\rightsquigarrow$ is a ternary relation over $D\times\widehat{R}\times\widehat{R}$
but as we have done with $\pitchfork$ and many other such relations
in this work, we drop the input argument when there is no confusion
and write it as a binary relation for readability. Before defining
semi-congruence, we introduce two concepts. One is the idea of \emph{useability}
of a space. A space $y$ is is useable, written $o^{*}(x,y)$, if
$\exists z.\,\chi(y,z)\wedge o(x,z)$, meaning a feasible solution
can be extracted from the space. The second is the notion of incorporating
sufficient information into a space to make it useable. This is defined
by an operator $\oplus:\widehat{R}\times t\rightarrow\widehat{R}$
that takes a space and some additional information of type $t$ and
returns a more defined space%
. The type $t$ depends on $\widehat{R}$. For example if $\widehat{R}$
is the type of lists, then $t$ might also be the same type. Now the
formal definition of semi-congruence is:
\[
y\rightsquigarrow y'\Rightarrow o^{*}(x,y'\oplus e)\Rightarrow o^{*}(x,y\oplus e)
\]
 That is, $y\rightsquigarrow y'$ is a sufficient condition for ensuring
that if $y'$ can be extended into a feasible solution than so can
$y$ \emph{with the same extension}. If $c$ is compositional (that
is, $c(s\oplus t)=c(s)+c(t)$) then it can be shown \cite{Ned12}
that if $y\rightsquigarrow y'$ and $y$ is cheaper than $y',$ then
$y$ \emph{dominates} $y'$ (written $y\vartriangleright y'$). Formally:
\begin{equation}
y\rightsquigarrow y'\wedge c(x,y)\leq c(x,y')\Rightarrow y\vartriangleright y'\label{eq:imm-dominance}
\end{equation}
The axioms given above extend GS-theory \cite{Smi88}.%

\begin{example}
\label{eg:dom-reln-for-SP-informal}Single Pair Shortest Path%
. If there are two paths $p$ and $p'$ leading from the start node,
if $p$ and $p'$ both terminate in the same node then $p\rightsquigarrow p'$.
The reason is that any path extension $e$ (of type $t=[Edge]$) of
$p'$ that leads to the target node is also a valid path extension
for $p$. Additionally if $p$ is shorter than $p'$ then $p$ dominates
$p'$, which can be discarded. Note that this does not imply that
$p$ leads to the target node, simply that no optimal solutions are
lost in discarding $p'$. This dominance relation is formally derived
in Eg. \ref{Dom-relation-for-SP}
\end{example}

\begin{example}
0-1 Knapsack

The 0-1 Knapsack problem is, given a set of items each of which has
a weight and utility and a knapsack that has some maximum weight capacity,
to pack the knapsack with a subset of items that maximizes utility
and does not exceed the knapsack capacity. %
Given combinations $k,k'$, if $k$ and $k'$ have both examined the
same set of items and $k$ weighs less than $k'$ then any additional
items $e$ that can be feasibly added to $k'$ can also be added to
$k$, and therefore $k\rightsquigarrow k'$. Additionally if $k$
has at least as much utility as $k'$ then $k\vartriangleright k'$.
\end{example}

The remaining sections %
{} cover the original contributions of this paper%
.

\section{A Theory Of Efficient Breadth-First Search (EBFS)}

While search can in principle solve for any computable function, it
still leaves open the question of how to carry it out effectively.
Various search strategies have been investigated over the years; two
of the most common being Breadth-First Search (BFS) and Depth-First
Search (DFS). It is well known that BFS offers several advantages
over DFS. Unlike DFS which can get trapped in infinite paths%
\footnote{resolvable in DFS with additional programming effort%
}, BFS will always find a solution if one exists. Secondly, BFS does
not require backtracking. Third, for deeper trees, BFS will generally
find a solution at the earliest possible opportunity%
. However, the major drawback of BFS is its space requirement which
grows exponentially. For this reason, DFS is usually preferred over
BFS.

Our first contribution in this paper is to refine GS-theory %
{} to identify the conditions under which a BFS algorithm can operate
space-efficiently. The key is to show how the size of the undominated
frontier of the search tree can be polynomially bounded. Dominance
relations are the basis for this.

In \cite{Smi88}, the relation $\pitchfork^{l}$ for $l\geq0$ is
recursively defined as follows:
\[
\begin{array}{rcl}
y\pitchfork^{0}y' & = & (y=y')\\
y\pitchfork^{l+1}y' & = & \exists y''\cdot\, y\pitchfork y''\wedge y''\pitchfork^{l}y'
\end{array}
\]
From this the next step is to define those spaces at a given frontier
level that are not dominated. However, this requires some care because
dominance is a pre-order, that is it satisfies the reflexivity and
transitivity axioms as a partial order does, but not the anti-symmetry
axiom. That is, it is quite possible for $y$ to dominate $y'$ and
$y'$ to dominate $y$ but $y$ and $y'$ need not be equal. An example
in Shortest Path is two paths of the same length from the start node
that end at the same node. Each path dominates the other. To eliminate
such cyclic dominances, %
define the relation $y\thickapprox y'$ as $y\vartriangleright y'\wedge y'\vartriangleright y$.
It is not difficult to show that $\thickapprox$ is an equivalence
relation. Now let the \emph{quotient frontier} at level $l$ be the
quotient set $\mathit{frontier_{l}}/\thickapprox$ . For type consistency,
let the \emph{representative} frontier $\mathit{rfrontier_{l}}$ $ $be
the quotient frontier in which each equivalence class is replaced
by some arbitrary member of that class. The representative frontier
is the frontier in which cyclic dominances have been removed. Finally
then the \emph{undominated} frontier $\mathit{undom_{l}}$ is $\mathit{rfrontier_{l}}-\{y\mid\exists y'\in\mathit{rfrontier_{l}}\cdot\, y'\vartriangleright y\}$.

Now given a problem in the GS class, if it can be shown that $\left\Vert undom_{l}\right\Vert $
for any $l$ is polynomially bounded in the size of the input, a number
of benefits accrue: (1) BFS can be used to tractably carry out the
search, as implemented in the raw program schema of Alg. \ref{alg:SEBFS-Pgm-Schema},
(2) The raw schema of Alg. \ref{alg:SEBFS-Pgm-Schema} can be transformed
into an efficient tail recursive form, in which the entire frontier
is passed down %
{} and (3) If additionally the tree depth can be polynomially bounded
(which typically occurs for example in \emph{constraint satisfaction
problems} or CSPs \cite{Dec03}) then, under some reasonable assumptions
about the work being done at each node, the result is a polynomial-time
algorithm for the problem.

\subsection{Program Theory}

A program theory for EBFS defines a recursive function which given
a space $y$, computes a non-trivial subset $F_{x}(y)$ of the optimal
solutions contained in $y$, where 
\[
F_{x}(y)=opt_{c}\{z\mid z\in y\wedge o(x,z)\}
\]
$opt_{c}$ is a subset of its argument that is the optimal set of
solutions (w.r.t. the cost function $c$), defined as follows:
\[
opt_{c}S=\{z\mid z\in S\wedge(\forall z'\in S\,\cdot\, c(z)\leq c(z'))\}
\]
Also let $undom(y)$ be $undom_{l(y)+1}\cap\{yy\mid y\pitchfork yy\}$
where $l(y)$ is the level of $y$ in the tree. The following proposition
defines a recurrence for computing the feasible solutions in a space:
\begin{prop}
\label{prop:G-recurrence}Let $\langle D,R,\widehat{R},o,c,\bot,\pitchfork,\chi,\vartriangleright,\oplus\rangle$
be a well-founded GS-Theory w.r.t. the subspace relation $\pitchfork$
and let $F_{x}(y)=\{z\mid z\in y\wedge o(x,z)\}$ be the set of feasible
solutions contained in $y$ and $G_{x}(y)=\{z\mid\chi(y,z)\wedge o(x,z)\}\cup\bigcup_{yy\pitchfork y}G_{x}(yy)\}$
be a recurrence. Then $G_{x}(y)=F_{x}(y)$ for any $y$\end{prop}
\begin{proof}
See \cite{Smi88}.
\end{proof}
Finally he following theorem defines a recurrence that can be used
to compute $FO_{x}(y)$: %

\begin{thm}
\label{thm:Soundness-of-GO}Let $\pitchfork$ be a well-founded relation
of GS-theory and let $GO_{x}(y)=opt_{c}\{z\mid\chi(y,z)\wedge o(x,z)\}\cup\bigcup_{yy\in undom(y)}GO_{x}(yy))$
be a recurrence. Then $GO_{x}(y)\subseteq FO_{x}(y)$\end{thm}
\begin{proof}
By generalized induction. The base case is those spaces which do not
have subspaces. Then $GO_{x}(y)=opt_{c}\{z\mid\chi(y,z)\wedge o(x,z)\}$.
By Prop. \ref{prop:G-recurrence} $ $$\{z\mid\chi(y,z)\wedge o(x,z)\}=\{z\mid z\in y\wedge o(x,z)\}$.
The inductive case is as follows: 
\[
\begin{array}{l}
FO_{x}(y)\\
=\{\mbox{defn}\}\\
opt_{c}\{z\mid z\in y\wedge o(x,z)\}\\
=\{\mbox{defn of }F_{x}\}\\
opt(F_{x}(y))\\
=\{F_{x}(y)=G_{x}(y)\mbox{ by Prop.}\ref{prop:G-recurrence}\}\\
opt(\{z\mid\chi(y,z)\wedge o(x,z)\}\cup\bigcup_{y\pitchfork yy}G_{x}(yy))\\
=\{G_{x}(yy)=F_{x}(yy)\mbox{ by Prop.}\ref{prop:G-recurrence}\}\\
opt(\{z\mid\chi(y,z)\wedge o(x,z)\}\cup\bigcup_{y\pitchfork yy}F_{x}(yy))\\
=\{\mbox{distributivity of }opt\}\\
opt(opt\{z\mid\chi(y,z)\wedge o(x,z)\}\cup opt(\bigcup_{y\pitchfork yy}F_{x}(yy)))\\
=\{\mbox{distributivity and idempotence of }opt\}\\
opt(\{z\mid\chi(y,z)\wedge o(x,z)\}\cup\bigcup_{y\pitchfork yy}opt(F_{x}(yy)))\\
=\{\mbox{unfold defn of }F_{x},\mbox{ fold defn of }FO_{x}\}\\
opt(\{z\mid\chi(y,z)\wedge o(x,z)\}\cup\bigcup_{y\pitchfork yy}FO_{x}(yy))\\
\supseteq\{yy\in undom(y)\Rightarrow y\pitchfork yy\}\\
opt(\{z\mid\chi(y,z)\wedge o(x,z)\}\cup\bigcup_{yy\in undom(y)}FO_{x}(yy))\\
\supseteq\{\mbox{induction hypothesis:}FO_{x}(yy)\supseteq GO_{x}(yy)\}\\
opt(\{z\mid\chi(y,z)\wedge o(x,z)\}\cup\bigcup_{yy\in undom(y)}GO_{x}(yy))\\
=\{\mbox{fold defn of }GO_{x}\}\\
GO_{x}(y)
\end{array}by
\]

\end{proof}
The theorem states that if the feasible solutions immediately extractable
from a space $y$ are combined with the solutions obtained from $GO_{x}$
of each undominated subspace $yy$, and the optimal ones of those
retained, the result is a subset of $FO_{x}(y)$. The next theorem
demonstrate non-triviality%
\footnote{Non-triviality is similar but not identical to completeness. Completeness
requires that \emph{every} optimal solution is found by the recurrence,
which we do not guarantee.%
} of the recurrence by showing that if a feasible solution exists in
a space, then one will be found.
\begin{thm}
Let $\pitchfork$ be a well-founded relation of GS-Theory and $GO_{x}$
be defined as above. Then
\[
FO_{x}(y)\neq\emptyset\Rightarrow GO_{x}(y)\neq\emptyset
\]
\end{thm}
\begin{proof}
The proof of Theorem \ref{thm:Soundness-of-GO} is a series of equalities
except for two steps. It is sufficient to show that both of these
steps preserve non-triviality. The proof is again by induction over
the subspace relation. The first refinement reduces $\bigcup_{y\pitchfork yy}FO_{x}(yy)$
to $\bigcup_{yy\in undom(y)}FO_{x}(yy)$. Suppose $\exists yy\cdot\, y\pitchfork yy\wedge FO_{x}(yy)\neq\emptyset$.
If $yy$$\in undom(y)$ then we are done. Otherwise if $yy$ is dominated,
then there is some $yy'\vartriangleright yy$ and by the property
of dominance, $FO_{x}(yy')\neq\emptyset$, so $\bigcup_{yy\in undom(y)}FO_{x}(yy)\neq\emptyset$.
The second refinement follows again by induction, using the induction
hypothesis $FO_{x}(yy)\neq\emptyset\Rightarrow GO_{x}(yy)\neq\emptyset$.
\end{proof}
From the characteristic recurrence we can straightforwardly derive
a simple recursive function \texttt{bfs} to compute a non-trivial
subset of $F_{x}$ for a given $y$, shown in Alg. \ref{alg:SEBFS-Pgm-Schema}
\begin{algorithm}
\begin{lyxcode}
\noindent {\scriptsize solve~::~}\textbf{\uline{\scriptsize D}}{\scriptsize{}~->~\{}\textbf{\uline{\scriptsize R}}{\scriptsize \}}{\scriptsize \par}

{\scriptsize solve(x)~=~bfs~x~\{}\uline{\scriptsize initial}{\scriptsize (x)\}}{\scriptsize \par}

{\scriptsize{}~}{\scriptsize \par}

{\scriptsize bfs~::~D~->~\{}\uline{\scriptsize RHat}{\scriptsize \}->~\{R\}}{\scriptsize \par}

{\scriptsize bfs~x~frontier~=}{\scriptsize \par}

{\scriptsize{}~~let~localsof~y~=~let~z~=~}\uline{\scriptsize extract}{\scriptsize{}~x~y~}{\scriptsize \par}

{\scriptsize{}~~~~~~~~~~~~~~~~~~~in~if~z!=\{\}~\&\&~}\uline{\scriptsize o}{\scriptsize (x,z)~then~z~else~\{\}}{\scriptsize \par}

{\scriptsize{}~~~~~~locals~=~(flatten.map)~localsof~frontier}{\scriptsize \par}

{\scriptsize{}~~~~~~allsubs~=~(flatten.map)~(subspaces~x)~frontier}{\scriptsize \par}

{\scriptsize{}~~~~~~undom~=~\{yy~:~yy$\in$allsubs~\&\&~}{\scriptsize \par}

{\scriptsize{}~~~~~~~~~~~~~~~~~~~~(yy'$\in$subs~\&\&~yy'~}\uline{\scriptsize `dominates`}{\scriptsize{}~yy~$\Rightarrow$~yy==yy')\}~}{\scriptsize \par}

{\scriptsize{}~~~~~~subsolns~=~bfs~x~undom}{\scriptsize \par}

{\scriptsize{}~~in~opt(locals~$\cup$~subsolns)}{\scriptsize \par}

{\scriptsize{}~}{\scriptsize \par}

{\scriptsize subspaces~::~D~->~RHat~->~\{RHat\}}{\scriptsize \par}

{\scriptsize subspaces~x~y~=~\{yy:~}\uline{\scriptsize split}{\scriptsize (x,y,yy))}{\scriptsize \par}

{\scriptsize{}~}{\scriptsize \par}

{\scriptsize opt~::~\{R\}~->~\{R\}}{\scriptsize \par}

{\scriptsize opt~zs~=~min~\{}\uline{\scriptsize c}{\scriptsize{}~x~z~|~z~$\in$zs\}}{\scriptsize \par}
\end{lyxcode}
\caption{\label{alg:SEBFS-Pgm-Schema}pseudo-Haskell Program Schema for EBFS
(schema parameters underlined)}
\end{algorithm}

The final program schema that is included in the Specware library
is the result of incorporating a number of other features of GS such
as necessary filters, bounds tests, and propagation, which are not
shown here. Details of these and other techniques are in \cite{Smi88}.

\subsection{A class of strictly greedy algorithms (SG)}

A greedy algorithm \cite{CLR01} is one which repeatedly makes a locally
optimal choice. For some classes of problems this leads to a globally
optimum choice. We can get a characterization of optimally greedy
algorithms within EBFS by restricting the size of $undom_{l}$ for
any $l$ to 1. If $undom_{l}\neq\emptyset$ then the singleton member
$y^{*}$ of $undom_{l}$ is called the \emph{greedy} choice. In other
work \cite{NSC12b} we show how to derive greedy algorithms for a
variety of problems including Activity Selection, One machine scheduling,
Professor Midas' Traveling Problem, Binary Search. %
{} %

\section{Methodology}

We strongly believe that every formal approach should be accompanied
by a methodology by which it can be used by a competent developer,
without needing great insights. Guided program synthesis already goes
a long way towards meeting this requirement by capturing design knowledge
in a reusable form. The remainder of the work to be done by a developer
consists of instantiating the various parameters of the program schema.
In the second half of this paper, we demonstrate how carrying this
out systematically allows us to derive several related graph algorithms,
revealing connections that are not always obvious from textbook descriptions.
We wish to reiterate that once the dominance relation and other operators
in the schema have been instantiated, \emph{the result is a complete
solution to the given problem}. We focus on dominance relations because
they are arguably the most challenging of the operators to design.
The remaining parameters can usually be written down by visual inspection.

The simplest form of derivation is to reason backwards from the conclusion
of $y\rightsquigarrow y'\Rightarrow o^{*}(x,y'\oplus e)\Rightarrow o^{*}(x,y\oplus e)$,
while assuming $o^{*}(x,y'\oplus e)$ %
. The additional assumptions that are made along the way form the
required semi-congruence condition. The following example illustrates
the approach. %

\begin{example}
\label{Dom-relation-for-SP}Derivation of the semi-congruence relation
for Single Pair Shortest Path in Eg. \ref{Spec-of-SP} is a straightforward
calculation as shown in Fig \ref{fig:semi-congruence-for-SP}. It
relies on the specification of Shortest Path given in Eg. \ref{Spec-of-SP}
and the GS-theory in Eg. \ref{GS-theory-for-SP}. %
{} 
\begin{figure}
$\begin{array}{l}
o^{*}(x,y\oplus e)\\
=\mbox{\{defn of }o^{*}\}\\
\exists z\cdot\,\chi(y\oplus e,z)\wedge o(x,z)\\
=\{\mbox{defn of }\chi\}\\
o(x,y\oplus e)\\
=\{\mbox{defn of }o\}\\
path?(y\oplus e,x.s,x.e)\\
=\{\mbox{distributive law for }path?\}\\
\exists n\cdot\, path?(y,x.s,n)\wedge path?(e,n,x.e)\\
\Leftarrow\{o^{*}(x,y'\oplus e)\mbox{, ie.}\exists m\cdot\, path?(y',x.s,m)\wedge path?(e,m,x.e).\mbox{ Let }m\mbox{ be witness for }n\}\\
path?(y,x.s,m)\wedge path?(e,m,x.e)\\
=\{m=last(y).t\mbox{, (where }last\mbox{ returns the last element of a sequence)}\}\\
last(y).t=last(y').t\wedge path?(y,x.s,n)
\end{array}$

\caption{\label{fig:semi-congruence-for-SP}Derivation of semi-congruence relation
for Single Pair Shortest Path}
\end{figure}

The calculation shows that a path $y$ is semi-congruent to $y'$
if $y$ and $y'$ both end at the same node and additionally $y$
is itself a valid path from the start node to its last node. Since
the cost function is compositional, this immediately produces a dominance
relation $y\vartriangleright y'=last(y)=last(y')\wedge path?(y,x.s,n)\wedge\sum_{edge\in y}edge.w\leq\sum_{edge'\in y'}edge'.w$.
Note the use of the distributive law for $path?$ in step 4. Such
laws are usually formulated as part of a domain theory during a domain
discovery process, or even as part of the process of trying to carry
out a derivation such as the one just shown. Given an appropriate
constructive prover (such as the one in KIDS \cite{Smi90}) such a
derivation could in fact be automated. Other examples that have been
derived using this approach are Activity Selection \cite{NSC10b},
Integer Linear Programming \cite{Smi88}, and variations on the Maximum
Segment Sum problem \cite{NC09}. 

While this dominance relation could in principle be used to computer
Single Source Shortest Path using a Best-First search (such as A{*})
it would not be very efficient as every pair of nodes on the frontier
would need to be compared. In the next section, a more powerful dominance
relation is derived which can be used in a Breadth-First search, and
even more importantly, be shown to be in the SG class, resulting in
a very efficient greedy algorithm. The dominance relation just derived
is still utilized, but in a subsidiary role. 
\end{example}

\section{More complex derivations: A family of related algorithms}

\subsection{(Single Source) Shortest Path}

Previously in Eg. \ref{Dom-relation-for-SP} we derived a dominance
relation for the (undirected) single pair shortest path problem. To
solve the problem of finding all shortest paths from a given start
node to every other node in the graph it is convenient to consider
the output as a set of edges that form what is called a \emph{path
tree}, a subgraph of the input graph which forms a spanning tree rooted
at the start node. The desired output is a path tree in which every
path from the root is the shortest. The specification of Single Pair
Shortest Path in Fig. \ref{fig:Spec-of-SP-1} is revised as shown
in Fig. \ref{fig:Spec-of-SP}
\begin{figure}
$\begin{array}{rcl}
D & \mapsto & \langle s:Node,edges:\{Edge\}\rangle\\
 &  & Edge=\langle a:Node,b:Node,w:Nat\rangle\\
R & \mapsto & \{Edge\}\\
o & \mapsto & \lambda(x,z)\cdot\, connected(x,z)\wedge acyclic(x,z)\\
c & \mapsto & \lambda(x,z)\,\cdot\sum_{p\in pathsFrom(x.s)}c'(p)\\
 &  & c'(p)=\sum_{edge\in p}edge.w
\end{array}$

\caption{\label{fig:Spec-of-SP}Specification of Shortest Path problem}

\end{figure}
The revised instantiation of Global Search theory is shown in Fig.
\ref{fig:GS-for-SP}
\begin{figure}
$\begin{array}{rcl}
\widehat{R} & \mapsto & R\\
\bot & \mapsto & \lambda x\cdot\,\{\}\\
\pitchfork & \mapsto & \lambda(x,p,pe)\cdot\,\exists e\in x.edges\cdot\, pe=p\cup\{e\}\\
\chi & \mapsto & \lambda(z,p)\cdot\, p=z\\
\oplus & \mapsto & \cup
\end{array}$

\caption{\label{fig:GS-for-SP}GS instantiation for Shortest Path}

\end{figure}
In what follows, the extends operator $\oplus$ is shown by simple
concatenation. The goal is to show that there is at most one undominated
child following a split of a partial solution $\alpha$. Let $\alpha e$
and $\alpha e'$ be two children following a split of $\alpha$, that
is the graphs $\alpha$ with edge $e$ added and that with $e'$ added.
Without loss of generality (w.l.o.g.) assume neither $e$ nor $e'$
are already contained in $\alpha$ and both connect to $\alpha$.
Let $z'=\alpha e'\omega'$ be a feasible solution derived from $\alpha e'$.
The task is to construct a feasible solution $z$ from $\alpha e$
and discover the conditions under which it is cheaper than $z'$.
We will use the definition of general dominance (\ref{eq:general-dominance}),
repeated here for convenience:
\[
\forall z'\in y'\cdot\, o(x,z')\Rightarrow\exists z\in y\cdot\, o(x,z)\wedge c(x,z)\leq c(x,z')
\]
Establishing $o(\alpha e\omega)$ requires \emph{connected}$(\alpha e\omega)$
and \emph{acyclic}$(\alpha e\omega)$. %

In guided program synthesis, it is often useful to write down laws
\cite{Smi90} that will be needed during the derivation. Some of the
most useful laws are distributive laws and monotonicity laws. For
example the following distributive law applies to $ayclic$ : 
\[
\mathit{acyclic}(\alpha\beta)=\mathit{acyclic}(\alpha)\wedge\mathit{acyclic}(\beta)\wedge\mathit{ac}(\alpha,\beta)
\]
where $ac$ defines what happens at the ``boundary'' of $\alpha$
and $\beta$: %
{} 
\[
\mathit{ac}(\alpha,\beta)=\forall m,n\cdot\,\exists p\in\alpha^{*}\cdot\, path?(p,m,n)\Rightarrow\neg\exists q\in\beta^{*}\cdot\, path?(q,m,n)
\]
requiring that if $p$ is a path in $\alpha$ ($\alpha^{*}$ is a
regular expression denoting all possible sequences of edges in $\alpha$)
connecting $m$ and $n$ then there should be no path between $m$
and $n$ in $\beta$. Examples of monotonicity laws are $acyclic(\alpha\beta)\Rightarrow acyclic(\alpha)$
and  $connected(\alpha) \wedge connected(\beta) \wedge nodes(\alpha)\cap nodes(\beta) \neq \emptyset  \Rightarrow connected(\alpha\beta$).
By using
the laws constructively, automated tools such as KIDS \cite{Smi90}
can often suggest instantiations for terms. For instance, after being told  $connected(\alpha e'\omega')$, the tool 
could apply the monotonicity
law for $connected$ to suggest $connected(\alpha ee'\omega')$,
that is to try $\omega=e'\omega'$. With this binding for $\omega$
we can attempt to establish $acyclic(\alpha e\omega)$, by expanding
its definition. One of the terms is $\ldots ac($$e,\alpha e'\omega')\ldots$
which fails because $conn(\alpha e'\omega')$ implies $path?(\alpha e'\omega',e.a,e.b)$
so adding edge $e$ creates a cycle. The witness to the failure is
some path $\pi'=e_{j}\ldots e_{k}$ where $e_{j}.a=e.a\wedge e_{k}.b=e.b$.
One possibility is that $e_{j}=e_{k}=e$, that is $\omega'$ contains
$e$. If so, let $\omega'=e\psi'$ for some $\psi'$. Then $z'=\alpha e'\omega'=\alpha e'e\psi'=\alpha ee'\psi'$.
Let $\omega=e'\psi'$ and now $z=z'$. Otherwise, w.l.o.g assume that
$e$ connects with $\alpha$ at $e.a$, and therefore so does $e_{j}$,
so the case $e_{j}.b=e.b$ is not very interesting. The only option
then is to remove edge $e_{k}$ from $\omega'$ . Let $\omega'=e_{k}\psi'$
and so $\omega$ is $e'\psi'$. Now the requirement becomes
\[
acyclic(\alpha e)\wedge acyclic(e'\psi')\wedge ac(\alpha e,e'\psi')
\]
 $acyclic(e'\psi')$ follows from $acyclic(\alpha e'\psi')$ by monotonicity
and $ac(\alpha e,e'\psi')$ follows from $acyclic(\alpha e'\psi')$
and the fact that $e$ was chosen above to remove the cycle in $\alpha ee'\omega'$
. This demonstrates $o(\alpha e\omega)$ provided $acyclic(\alpha e)$
where $\omega$ is constructed as above. 

Finally, to establish general dominance it is necessary to show that
$z$ costs no more than $z'$. Note that the cost of a solution is
the sum of individual path costs starting from $x.s$. Let $m$ denote
$e.a$ and $n$ denote $e.b$ (and analogously for $e'$). Now consider
a path to a node $p$ in $z'$. If the path to $p$ does not contain
edge $e_{k}$ ie. pass through $n$ then the same path holds in $z$.
Otherwise let $\beta'e'_{i}\gamma'e"\delta$ be a path to $p$ in
$z'$ where $e"$ is the edge $e_{k}$ above (see Fig. \ref{fig:APSP-y'}).
\begin{figure}
\includegraphics[scale=0.7]{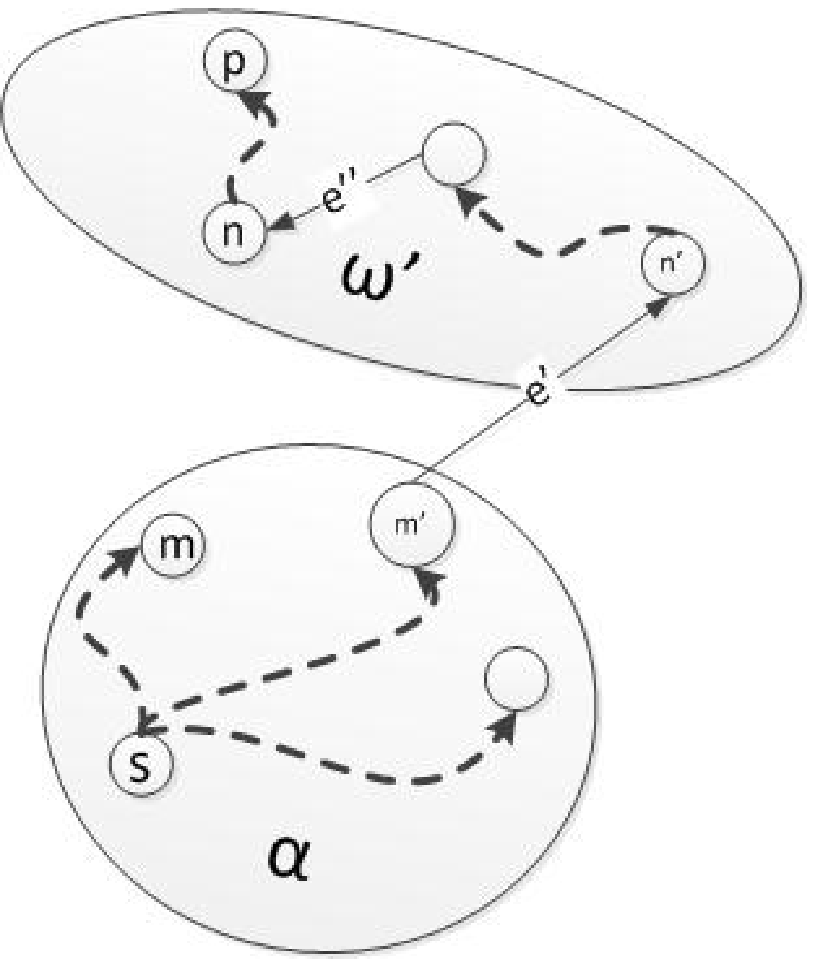}\caption{\label{fig:APSP-y'}Feasible solution $\alpha e'\omega'$}
\end{figure}
 $\beta'$ is a path from $x.s$ in $\alpha$ and $e'_{i}$ is some
edge that leads out of $\alpha$ on the path to $p$. Then the corresponding
path in $z$ is $\beta e\delta$ (see Fig. \ref{fig:APSP-y}). 
\[
\begin{array}{l}
c(\alpha e\omega,\beta e\delta)\leq c(\alpha e'\omega',\beta'e{}_{i}\gamma'e"\delta)\\
=\{\mbox{expand defns}\}\\
c(z,\beta e)+c(z,\delta)\leq c(z',\beta'e_{i})+c(z',e")+c(z,\delta)\\
=\{\mbox{+ve edge weights, triangle inequality}\}\\
c(z,\beta e)\leq c(z',\beta'e{}_{i})
\end{array}
\]
As there were no restrictions on $e'$ above, let $e_{i}$ be the
witness for $e'$ and this establishes
\begin{figure}
\includegraphics[scale=0.7]{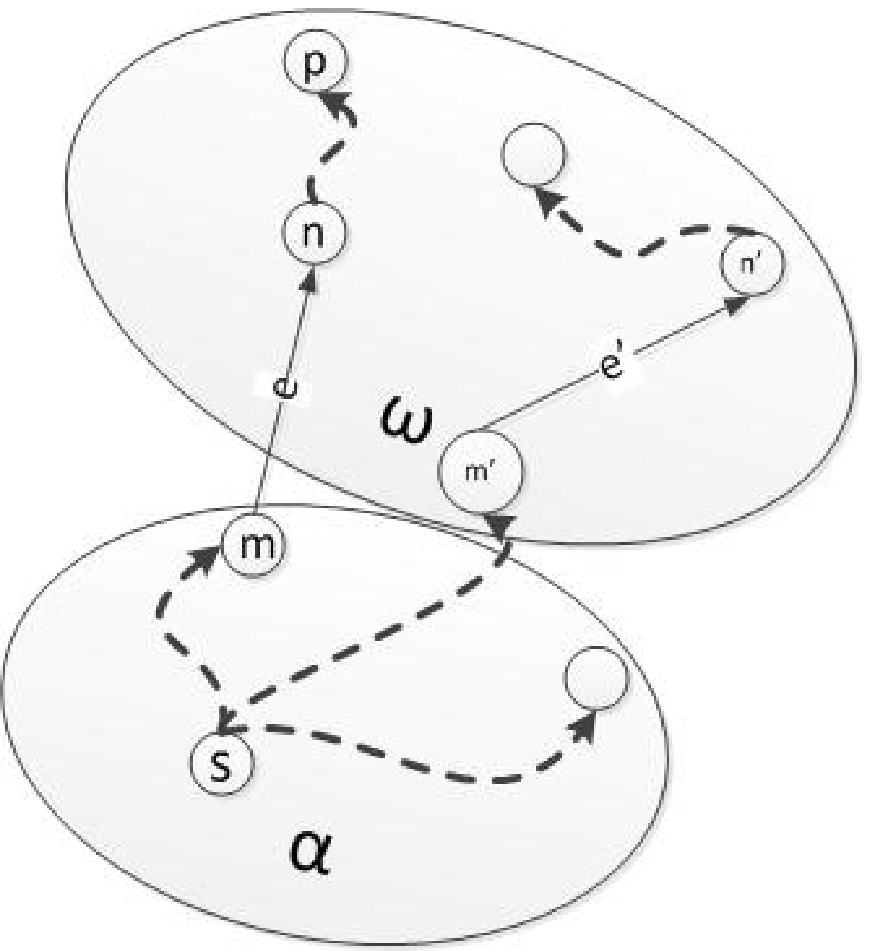}

\caption{\label{fig:APSP-y}Feasible solution $\alpha e\omega$}
\end{figure}
That is, provided the path $\beta e$ is shorter than $\beta'e'$,
there cannot be a shorter path via $e'$ to $n$. As cost is the sum
of the path costs, it follows that $c(x,\alpha e\omega)\leq c(\alpha e'\omega')$.
The dominance condition is then $\alpha e\vartriangleright\alpha e'\Leftarrow acyclic(\alpha e)\wedge c(z,\beta e)\leq c(z',\beta'e')$
. Finally, as %
{} it is not known at the time of the split which $e'$ will lie on
the path to $e.b$, to be conservative, let $e$ be that edge whose
endpoint $e.b$ is the closest to the start node. This is therefore
the greedy choice. Iincorporating finite differencing to incrementally
maintain the distances to the nodes, using the dominance relation
derived earlier for Single Pair Shortest Path to eliminate the longer
paths to a node, and data structure refinement results in an algorithm
similar to Dijkstra's algorithm for MSTs.

\subsection{Minimum Spanning Tree}

The specification of MST is very similar to that of Shortest Path,
with the difference that there is no longer a distinguished node $s$,
the input graph must be connected, and the cost of a solution is simply
the sum of the weights of the edges in the tree
\begin{figure}
$\begin{array}{rcl}
D & \mapsto & \{Edge\}\mid connected\\
 &  & Edge=\langle f:Node,t:Node,w:Nat\rangle\\
R & \mapsto & \{Edge\}\\
o & \mapsto & \lambda(x,z)\cdot\, connected(x,z)\wedge acyclic(x,z)\\
c & \mapsto & \lambda(x,z)\,\cdot\sum_{edge\in z}edge.w
\end{array}$

\caption{\label{fig:Spec-of-MST}Specification of Min. Spanning Tree problem}

\end{figure}
The instantiation of the simple GS operators is as for SP. Again,
any edge that is added must not create a cycle or it cannot lead to
a feasible solution. We will describe the algorithm construction process
informally so as to expose the connection with the Shortest Path algorithm
more clearly. However the derivations shown here can also be similarly
formalized.

There are now two ways to satisfy the acyclicity requirement. One
is by choosing an edge connecting a node in $\alpha$ to one outside
of $\alpha$. Another is to choose an edge that connects two nodes
within $\alpha$, being careful not to create cycles. The two options
are examined next,

Option 1: Let $z'=\alpha e'\omega'$ be a feasible solution derived
from $\alpha e'$. If $\omega'$ includes $e$ then let $\omega$
in a feasible solution $z=\alpha e\omega$ simply be $\omega'-\{e\}\cup\{e'\}$
and then $z=z'$. Otherwise, if $\omega$' does not contain $e$ there
must be some other path connecting $\alpha$ with $e.t$. W.l.o.g.
assume that path is via $e'$. If $\alpha e'\omega'$ is feasible,
then it is a tree, so $\omega$' is also a tree. Therefore it is not
difficult to show that $z=\alpha e\omega'$ is also a spanning tree.
Now to show dominance, derive conditions under which $z$ is cheaper
than $z'$: 
\[
\begin{array}{l}
c(x,\alpha e\omega')\leq c(x,\alpha e'\omega')\\
=\{\mbox{defn of }c\}\\
\sum_{edge\in\alpha e\omega'}edge.w\leq\sum_{edge\in\alpha e'\omega'}edge.w\\
=\\
e.w\leq e'.w
\end{array}
\]
Finally, as %
{} it is not known at the time of the split which $e'$ will lie on
the path to $e.t$, to be conservative, let $e$ be that edge with
the least weight connecting $\alpha$ with an external node . This
is therefore the greedy choice. The result is an algorithm that is
similar to Prim's algorithm for MSTs.

Option 2: The difference with Option 1 is in how $e$ is chosen in
order to ensure acyclicity. For a feasible solution, $\alpha$ must
not contain any cycles. Therefore it consists of a collection of acyclic
connected components, ie trees. Any new edge cannot connect nodes
within a component without introducing a cycle. Therefore it must
connect two component trees. Connecting two trees by a single edge
results in a new tree. As in Option 1, let $z'=\alpha e'\omega'$
be a feasible solution derived from $\alpha e'$. If $\omega'$ includes
$e$ then let $\omega$ in a feasible solution $z=\alpha e\omega$
simply be $\omega'-\{e\}\cup\{e'\}$ and then $z=z'$. Otherwise,
if $\omega$' does not contain $e$ there must be some other edge
used to connect the two trees that $e$ would have connected. W.l.o.g.
assume that edge is $e'$. If $\alpha e'\omega'$ is feasible, then
it is a tree, so $\omega$' is also a tree. Therefore it is not difficult
to show that $z=\alpha e\omega'$ is also a spanning tree. The derivation
of a cost comparison relation is identical to Option 1, and once again
the greedy choice is the edge $e$ that connects two trees and is
of least weight. The result of this option is an algorithm that is
similar to Kruskal's algorithm.

In conclusion, we observe that far from being completely different
algorithms, Dijkstra's algorithm, Prim's algorithm and Kruskal's algorithm
differ only in very small number of (albeit important) ways. In contrast,
many textbook descriptions of the algorithms introduce the algorithms
out of the blue, followed by separate proofs of correctness. We have
shown how a systematic procedure can derive different algorithms,
with relatively minor changes to the derivations.

\section{Related Work}

Gulwani et al. \cite{SGF10,GJTV11} describe a powerful program synthesis
approach called \emph{template-based synthesis}. A user supplies a
template or outline of the intended program structure, and the tool
fills in the details. A number of interesting programs have been synthesized
using this approach, including Bresenham's line drawing algorithm
and various bit vector manipulation routines. A related method is
inductive synthesis \cite{IGMS10} in which the tool synthesizes a
program from examples. The latter has been used for inferring spreadsheet
formulae from examples. All the tools rely on powerful SMT solvers.
The Sketching approach of Solar-Lezama et al \cite{PBS11} also relies
on inductive synthesis. A \emph{sketch}, similar in intent to a template,
is supplied by the user and the tool fills in such aspects as loop
bounds and array indexing. Sketching relies on efficient SAT solvers.
{} To quote Gulwani et al. the benefit of the template approach is that
``the programmer only need write the structure of the code and the
tool fills out the details'' \cite{SGF10}.%
Rather than the programmer supplying an arbitrary template, though,
we suggest the use of a program schema from the appropriate algorithm
class (refer to Step 2 of the process in Sec. \ref{sub:Process-1}).
We believe that the advantage of such an approach is that, based on
a sound theory, much can already be inferred at the abstract level
and this is captured in the theory associated with the algorithm class.
Furthermore, knowledge of properties at the abstract level allows
specialization of the program schema with information that would otherwise
have to either be guessed at by the programmer devising a template
or inferred automatically by the tool (e.g. tail recursive implementation
or efficient implementation of dominance testing with hashing). We
believe this will allow semi-automated synthesis to scale up to larger
problems such as constraint solvers (SAT, CSP, LP, MIP, etc.), planning
and scheduling, and O/S level programs such as garbage collectors
\cite{PPS10}.

Program verification is another field that shares common goals with
program synthesis - namely a correct efficient program. The difference
lies in approach - we prefer to construct the program in a way that
is guaranteed to be correct, as opposed to verifying its correctness
after the fact. Certainly some recent tools such as Dafny \cite{Lei10}
provide very useful feedback in an IDE during program construction.
But even such tools requires significant program annotations in the
form of invariants to be able to automatically verify non-trivial
examples such as the Schorr-Waite algorithm \cite{Lei10}. Nevertheless,
we do not see verification and synthesis as being necessarily opposed.
For example, ensuring the correctness of the instantiation of several
of the operators in the program schema which is usually done by inspection
is a verification task, as is ensuring correctness of the schema that
goes in the class library. We also feel that recent advances in verification
via SMT solvers will also help guided synthesis by increasing the
degree of automation.

Refinement is generally viewed as an alternative to synthesis. A specification
is gradually refined into an efficient executable program. Refinement
methods such as Z and B have proved to be very popular. In contrast
to refinement, guided program synthesis already has the program structure
in place, and the main body of work consists of instantiating the
schema parameters followed by various program transformations many
of which can be mechanically applied. Both refinement and synthesis
rely extensively on tool support, particularly in the form of provers.
We expect that advances in both synthesis and refinement will benefit
the other field.%

\section{Summary and Future Work}

We have formulated a theory of efficient breadth-first search based
on dominance relations. A very useful specialization of this class
occurs when there is at most one undominated child node. This is the
class of Strictly Greedy algorithms. We have also derived a recurrence
from which a simple program schema can be easily constructed. We have
shown how to systematically derive dominance relations for a family
of important graph algorithms revealing connections between them that
are obscured when each algorithm is presented in isolation. 

Nearly all the derivations shown in this paper have been carried out
by hand. However, they are simple enough to be automated. We plan
on building a prover that incorporates the ideas mentioned in here.
We are encouraged by the success of a similar prover that was part
of KIDS, a predecessor to Specware.

\bibliographystyle{eptcs}
\bibliography{../my}

\end{document}